\documentclass{article}
\usepackage{dsfont}
\usepackage[ngerman,english]{babel}
\usepackage[inline]{enumitem}
\usepackage{amsmath,amsfonts,amsthm}

\DeclareMathOperator{\sgn}{sgn}
\DeclareMathOperator{\tr}{Tr}

\newtheorem{theorem}{Theorem}
\newtheorem{definition}{Definition}
\newtheorem{lemma}{Lemma}
\newtheorem{remark}{Remark}
\newtheorem{proposition}{Proposition}

\newcommand{\proofparagraph}[1]{\medskip\noindent{\bf #1}}

\newcommand{\dda}{{\rm d}}
\newcommand{\dd}{\,\dda}
\newcommand{\ddd}[1]{\dd^{#1}}
\newcommand{\ee}{{\rm e}}
\newcommand{\F}{\mathcal{F}}

\begin{document}

\title{On the BCS gap equation for superfluid fermionic gases
}

\author{\small Gerhard Br\"aunlich$^1$ and Christian Hainzl$^2$\\
\small\it Mathematical Institute, University of T{\"u}bingen \\[-1mm]
  \small\it Auf der Morgenstelle 10, 72076 T\"ubingen, Germany \\[-1mm]
  \small\it $^1$E-mail: gerhard.braeunlich@uni-tuebingen.de\\[-1mm]
  \small\it $^2$E-mail: christian.hainzl@uni-tuebingen.de \\\\
  \small Robert Seiringer \\
  \small\it Institute of Science and Technology Austria\\[-1mm]
  \small\it Am Campus 1, 3400 Klosterneuburg, Austria\\[-1mm]
  \small\it E-mail: robert.seiringer@ist.ac.at}

\date{\today}

\maketitle

\begin{abstract}
  We present a rigorous derivation of the BCS gap equation for superfluid fermionic gases with point interactions. Our starting point is the BCS energy functional, whose minimizer we investigate in the limit when the range of the interaction potential goes to zero.
\end{abstract}


\section{Introduction}
\label{sec:introduction}

In BCS theory \cite{bcs, Leggett, NRS} of superfluid fermionic gases the interaction 
between the spin-$1/2$ fermions is usually modeled by a contact potential. 
This approximation is justified for the low density atomic gases usually observed in the lab, since the range of the effective interaction is much smaller than the 
mean particle distance.  In the theoretical physics literature \cite{Leggett,randeria,NRS} the states of superfluidity are usually characterized via the simplified  BCS gap equation
\begin{equation}\label{se}
  -\frac{1}{4\pi a} =
  \frac{1}{(2\pi)^3}\int_{\mathbb{R}^3}\left(\frac{\tanh \big( \frac{ \sqrt{(p^2 - \mu)^2 + |\Delta|^2}}{2T} \big)}{\sqrt{(p^2 - \mu)^2 + |\Delta|^2}}
    -\frac{1}{p^2}\right) \ddd{3}p\,,
\end{equation}
where $\mu$ is a fixed chemical potential and $\Delta$ is the corresponding order parameter of the system.
This order parameter does not vanish below the  critical temperature $T_c$, uniquely defined by 
 $$  -\frac{1}{4\pi a} = \frac{1}{(2\pi)^3}\int_{\mathbb{R}^3}
    \left(
      \frac{\tanh\big(\frac{p^2-\mu}{2T_c}\big)}{p^2-\mu}
      -\frac{1}{p^2} \right) \ddd{3}p
$$
for $a<0$. 
The parameter $a$ is the scattering length of the corresponding interaction and the usual argument for the derivation of  equation \eqref{se} involves an ad-hoc renormalization scheme.
The main goal of this paper is to give a rigorous derivation of
equation \eqref{se} starting from the BCS functional of superfluidity \cite{Leggett,NRS,HHSS} for a sequence of interaction potentials 
$V_\ell$ with range tending to zero, and  scattering length $a(V_\ell)$ converging to a negative $a$. 
The present result is a consequence of our previous work \cite{BHS} where we treated the more general case of BCS-Hartree-Fock theory and where we allowed for potentials with strong repulsive core. Dropping the direct and exchange term, however, as we do here, allows us to give a shorter,
more transparent derivation of \eqref{se}, and also to work with simpler assumptions on the interaction potentials.


\section{The Model}
\label{sec:model}

We consider a gas of spin $1/2$ fermions in the thermodynamic limit at
temperature $T \geq 0 $ and chemical potential $\mu \in
\mathbb{R}$. The particles interact via a local two-body potential
which we denote by $V$. The state of the system is described by two
functions $\hat\gamma:\mathbb{R}^3\to \mathbb{R}_+$ and
$\hat\alpha:\mathbb{R}^3\to \mathbb{C}$, which are conveniently
combined into a $2\times 2$ matrix 
\begin{equation}
  \label{def:Gamma}
  \Gamma(p) = \left(
    \begin{array}{cc}
      \hat{\gamma}(p) & \hat{\alpha}(p)\\
      \overline{\hat{\alpha}(p)} & 1-\hat{\gamma}(-p)
    \end{array}
  \right),
\end{equation}
required to satisfy $0\leq \Gamma \leq \mathds{1}_{\mathbb{C}^2}$ at
every point $p\in\mathbb{R}^3$. The function $\hat\gamma$ is
interpreted as the momentum distribution of the gas, while $\alpha$
(the inverse Fourier transform of $\hat\alpha$) is the Cooper pair
wave function. Note that there are no spin variables in $\Gamma$; the
full, spin dependent Cooper pair wave function is the product of
$\alpha(x-y)$ with an antisymmetric spin singlet.

The \emph{BCS functional} $\mathcal{F}_T^V$, whose infimum over all
states $\Gamma$ describes the negative of the pressure of the system,
is given as 
\begin{equation}
  \label{eq:F_T}
  \begin{split}
    \mathcal{F}_T^V(\Gamma) =& \int_{\mathbb{R}^3} (p^2 -
    \mu)\hat{\gamma}(p) \ddd{3}p +\int_{\mathbb{R}^3} |\alpha(x)|^2
    V(x) \ddd{3}x - T S(\Gamma),
  \end{split}
\end{equation}
where
\begin{equation*}
  S(\Gamma) = -\int_{\mathbb{R}^3} \tr_{\mathbb{C}^2} \big(\Gamma(p) \ln \Gamma(p)\big) \ddd{3}p
\end{equation*}
is the entropy of the state $\Gamma$. The functional \eqref{eq:F_T}
can be obtained by restricting the many-body problem on Fock space to
translation-invariant and spin-rotation invariant quasi-free states, and dropping the direct and exchange term in the interaction energy, 
see \cite[Appendix A]{HHSS} and \cite{BLS}.

The \emph{normal state} $\Gamma_0$ is the minimizer of the functional \eqref{eq:F_T}
restricted to states with $\alpha=0$. It is given by
\begin{equation*}
  \hat{\gamma}_0(p) = 
  \frac{1}{1+\ee^{\frac{p^2 - \mu}{T}}}.
\end{equation*}
The system is said to be in a 
superfluid phase if and only if the minimum of $\mathcal{F}_T^V$ is
not attained at a normal state,  and we call a normal state $\Gamma_0$
\emph{unstable} in this case. 

\bigskip 

In a previous work \cite{HHSS} we thoroughly studied the functional \eqref{eq:F_T}.
It is not difficult to see that this functional has a (not necessarily unique) minimizer $(\gamma,\alpha)$.
More difficult is the question under which circumstances it is possible to guarantee that $\alpha$ does not vanish.
Such a non-vanishing $\alpha$ in fact describes a macroscopic coherence of pairs such that the system displays a superfluid behavior. 
The corresponding Euler-Lagrange equations for $\gamma$ and $\alpha$ can 
be equivalently expressed via $\Delta = 2 (2\pi)^{-3/2} \hat V \ast \hat \alpha$ in the form of the BCS {\em gap equation}
\begin{equation}\label{bcseintro}
\Delta(p) = -\frac 1{(2\pi)^{3/2}} \int_{\mathbb{R}^3} \hat V(p-q)
\frac{\Delta(q)}{E_\mu^\Delta(q)} \tanh \frac{E_\mu^\Delta(q)}{2T} \, \ddd{3}q
\end{equation}
with $E_\mu^\Delta(p)= \sqrt{(p^2-\mu)^2 +
|\Delta(p)|^2}$; here, $\hat V$
denotes the Fourier transform of $V$. The function $\Delta(p)$ is
the order parameter and is related to the wavefunction of the {\em
Cooper pairs}. The equation
\eqref{bcseintro} is highly non-linear; nonetheless,  it is possible 
to show 
 \cite{HHSS} that the existence of a non-trivial solution to
\eqref{bcseintro} at some temperature $T$ is equivalent to the fact
that a certain {\em linear operator} has a negative eigenvalue. For
$T=0$ this operator is given by the Schr\"odinger-type operator
$|-\Delta - \mu| +  V$. This rather astonishing fact that one
can reduce a non-linear to a linear problem, allowed for a more
thorough mathematical study. Using spectral-theoretic methods, the class of potentials leading
to a non-trivial solution for \eqref{bcseintro} has  been 
precisely characterized. For instance, in \cite{FHNS2007} it was shown 
that if  $\int V(x) dx < 0$, then there exists a
critical temperature $T_c(V) > 0$ such that
\eqref{bcseintro} attains a non-trivial solution for all $T <
T_c(V)$, whereas there is no solution for $T \geq
T_c(V)$. Additionally, in \cite{FHNS2007} the
precise asymptotic behavior of $T_c(\lambda V)$ in the small
coupling limit $\lambda \to 0$ was determined; the resulting expression generalizes well-known formulas in the physics literature \cite{Gorkov,NRS} valid only at low density. 
The low density limit $\mu\to 0$ of the critical temperature was studied in \cite{HS-mu}.

%
%

\section{Main Results}
\label{sec:results}

We study the case of short-range interaction
potentials $V_\ell$, with range $\ell$ tending to zero in such a way that $V_\ell$ converges to a contact interaction. Such contact interactions are thoroughly studied in the literature \cite[chap
I.1.2-4]{albeverio} and are known to arise as a one parameter family of self-adjoint extensions of the Laplacian on $\mathbb{R}^3\setminus\{0\}$. The relevant parameter uniquely determining the extension is, in fact, the scattering length, which we assume to be negative, in which case the resulting operator is non-negative, i.e., there are no bound states. In other words, we require that the  scattering length $a(V_\ell)$ converges to a negative value as $\ell \to 0$, i.e., 
$$\lim_{\ell \to 0} a(V_{\ell}) = a < 0.$$ It was pointed out in  
\cite[Equ. (3)]{HS-mu} that the scattering length of any potential $V\in L^1 \cap L^{3/2}$ can be written as
\begin{equation}
  \label{eq:a}
  a(V) = \frac{1}{4\pi} \left< |V|^{1/2} \middle| \tfrac{1}{1+V^{1/2}
    \frac{1}{p^2}|V|^{1/2}} V^{1/2}\right>
\end{equation}
where $V^{1/2}$ is defined by $V^{1/2}= V |V|^{-1/2}$. 
Note that in case of a two-body interaction that does allow bound states the system would display 
features of a Bose-Einstein condensate of fermion pairs in the low density limit, see \cite{randeria, zwerger-1992, Pieri-Strinati,HS, HSch}. 

In order to obtain a non-vanishing limit of the sequence of scattering lengths $a(V_\ell)$, we have to adjust the sequence of potentials so that 
the corresponding Schr\"odinger operator  just barely fails to have a bound state. 
We shall follow the method of \cite[chap
I.1.2-4]{albeverio} and first choose a potential $V$ such that $p^2 + V$ is
non-negative and has a simple
zero-energy resonance. Equivalently this means that  the corresponding Birman-Schwinger operator $$V^{1/2}\tfrac{1}{p^2}|V|^{1/2}$$ has $-1$ as lowest, simple, eigenvalue. 
Next we scale this potential and multiply it by a factor $\lambda(\ell) < 1$, such that the corresponding $V_\ell$ 
does no longer have a zero resonance, but a negative scattering length. 
Recall that a potential with zero resonance has an infinite scattering length. 

To be precise, we define $V_\ell$ according to
\begin{equation}
  \label{eq:scaling}
  V_\ell(x) = \lambda(\ell) \ell^{-2} V(\tfrac{x}{\ell}),
\end{equation}
where
$\lambda(0) = 1$, $\lambda < 1$ for all $\ell>0$ and
$1-\lambda(\ell) = O(\ell)$.  
We are interested in the limit $\ell \to 0$ meaning that the range of the potential converges to zero.  
This scaling essentially leaves the $L^{3/2}$ norm of $V$
invariant, but  the $L^1$ norm vanishes linearly in $\ell$.
This is a major difference to the work 
\cite{BHS} where we allowed the point interaction to  be approximated by a
sequence $V_\ell$, whose $L^1$ norm converges to a positive number, with its $L^{3/2}$-norm even diverging. 
This required a new approach in the proof and led to a more general statement about contact interactions \cite{BHS2}. 
In the case considered here 
it suffices to rely on results of \cite[chap
I.1.2-4]{albeverio}.

Our main objective now is to consider the solution $\Delta_\ell$ of the BCS gap-equation \eqref{bcseintro}, coming from a minimizer of the functional ${ \mathcal{F}}_T^{V_\ell}$,
and to show that in the limit where the range $\ell$ goes to zero, i.e., the potentials $V_\ell$ tend to a contact interaction, 
the order parameter $\Delta_\ell$ converges to a constant function $\Delta$, which satisfies the simplified equation \eqref{se}. 
This equation appears throughout the physics literature as the one describing superfluid systems. 

\medskip 

It is obvious that the critical temperature $T_c$ is defined as the temperature where $\Delta=0$ satisfies the equation \eqref{se}.
More precisely:
\begin{definition}[Critical temperature]
  \label{def:tc}
  Let $\mu >0$.  The \emph{critical temperature} $T_c$ corresponding to the scattering length $a<0$ is given by the
  equation
  \begin{equation}
    \label{eq:T_c}
    -\frac{1}{4\pi a} = \frac{1}{(2\pi)^3}\int_{\mathbb{R}^3}
    \left(
      \frac{\tanh\big(\frac{p^2-\mu}{2T_c}\big)}{p^2-\mu}
      -\frac{1}{p^2} \right) \ddd{3}p.
  \end{equation}
\end{definition}
Since the function
$$ T \mapsto  \int_{\mathbb{R}^3}
    \left(
      \frac{\tanh\big(\frac{p^2-\mu}{2T}\big)}{p^2-\mu}
      -\frac{1}{p^2} \right) \ddd{3}p
$$
is strictly monotone in $T$ the critical temperature $T_c$ is unique. 

As our main theorem we reproduce the BCS gap equation 
for contact interactions, see, e. g., \cite[Eq.~(10)]{Leggett}, \cite[Eq.~(7)]{randeria}.
\begin{theorem}[Effective Gap equation]
  \label{thm:gap_eff}
  Let $T \geq 0$, $\mu\in \mathbb{R}$ and assume $V \in L^{3/2}({\mathbb{R}^3}) \cap  L^{1}({\mathbb{R}^3})$ and $|x| V(x) \in  L^1({\mathbb{R}^3}) \cap L^{2}({\mathbb{R}^3}) $. Let further
  $(\hat{\gamma}_\ell,\hat{\alpha}_\ell)$ be a minimizer of
  $\mathcal{F}_T^{V_\ell}$ with corresponding $\Delta_\ell = 2
  (2\pi)^{-3/2}\hat V_\ell * \hat\alpha_\ell$. Then there exist
  $\Delta \geq 0$ such that
  $|\Delta_\ell(p)| \to \Delta$ pointwise as $\ell \to 0$.
  If $\Delta\neq 0$ then it satisfies the equation 
  \begin{equation}
    \label{eq:gap_eff}
    \boxed{
      -\frac{1}{4\pi a} =
      \frac{1}{(2\pi)^3}\int_{\mathbb{R}^3}\left(\frac{1}{K_{T,\mu}^{\Delta}}
        -\frac{1}{p^2}\right) \ddd{3}p\,,
    }
  \end{equation}
  where we use the abbreviations
\begin{equation*}
  K_{T,\mu}^{\Delta}(p) =
  \frac{E_{\mu}^{\Delta}(p)}{\tanh\big(\frac{E_{\mu}^{\Delta}(p)}{2T}\big)}\,, \quad 
  E_{\mu}^{\Delta}(p) = \sqrt{(p^2 -
    \mu)^2 + |\Delta|^2}\,.
 \end{equation*}
Furthermore,  the limiting  $\Delta$ does not vanish if and only if $T < T_c$.
 \end{theorem}



\section{Proofs}
\label{sec:proofs}


Recall that we chose $V$ so that $V^{1/2}\tfrac{1}{p^2}|V|^{1/2}$ has $-1$ as lowest simple eigenvalue, i.e., 
there is a unique $\phi$, with 
$$ \left( V^{1/2}\tfrac{1}{p^2}|V|^{1/2} + 1\right) \phi =0.$$ 
With 
 $U_\ell$ denoting the unitary operator $(U_\ell\varphi)(x) =
\ell^{-3/2}\varphi(\frac{x}{\ell})$,
we can rewrite $$ V_\ell (x) = \lambda(\ell) \ell^{-2} V(\tfrac{x}{\ell}) = \lambda(\ell) \ell^{-2}  U_\ell V U^{-1}_\ell. $$
Since $U_\ell p^2 U^{-1}_\ell = \ell^2 p^2$, it is easy to see that 
$$ U_\ell V^{1/2} \frac 1 {p^2} |V|^{1/2} U_\ell^{-1} = \frac{1}{\lambda(\ell)} V_\ell^{1/2} \frac 1 {p^2} |V_\ell|^{1/2}.$$
Denoting $\phi_\ell = U_\ell \phi$, this implies 
\begin{equation}\label{Xell}  V_\ell^{1/2} \frac 1 {p^2} |V_\ell|^{1/2} \phi_\ell = \lambda(\ell) U_\ell V^{1/2} \frac 1 {p^2} |V|^{1/2} \phi = - \lambda(\ell) \phi_\ell.
\end{equation}
showing that the lowest eigenvalue of 
$1 + V_\ell^{1/2} \frac 1 {p^2} |V_\ell|^{1/2}$ is $1 - \lambda(\ell) = O(\ell)$.  
Moreover, note that 
\begin{equation}
  \label{eq:V:p}
  \|V_\ell\|_p = \lambda(\ell) \ell^{3/p-2}\| V\|_p
\end{equation}
for $p\geq 1$. 
\begin{remark}
  In \cite[Appendix A.1]{BHS}, we show that the scattering length corresponding to
  the potential $V_\ell$, indeed, converges to the negative value
$$
  \lim_{\ell\to 0}a(V_\ell) = -\frac{1}{\lambda'(0)} \frac{|\langle
    |V|^{1/2}| \phi\rangle|^2}{\langle \sgn(V) \phi|\phi\rangle} < 0.$$
\end{remark}

In the next Lemma, we derive a lower bound for the BCS functional which is uniform in $\ell$.
This will allow us to obtain limits for the order parameter $\Delta_\ell$.

\begin{lemma}
  \label{lemma:minimizer}
  There exists $C_1 > 0$, independent of $\ell$,  such that 
  \begin{equation}
    \label{eq:F_T_bound}
    \mathcal{F}_T^{V_\ell}(\Gamma)
    \geq -C_1
    + \frac{1}{2}\int_{\mathbb{R}^3} (1+p^2)(\hat{\gamma} - \hat{\gamma}_0)^2 \ddd{3}p
    + \frac{1}{2}\int_{\mathbb{R}^3} |p|^b |\hat{\alpha}|^2 \ddd{3}p,
  \end{equation}
  for $0 \leq b < 1$, 
  where we denote $\hat{\gamma}_0(p) =
  \frac{1}{1+\ee^{(p^2-\mu)/T}}$.
\end{lemma}

\begin{proof}
  For details of the proof of this  Lemma we refer to \cite[Lemma 3]{BHS}.  
  The main observation in the proof is that we may express the difference
  $\mathcal{F}_T^{V_\ell}(\Gamma) - \mathcal{F}_T^{V_\ell}(\Gamma_0)$ as
  \begin{equation*}
    \mathcal{F}_T^{V_\ell}(\Gamma) -
    \mathcal{F}_T^{V_\ell}(\Gamma_0)
    = \frac{T}{2}\mathcal{H}(\Gamma,\Gamma_0) + \int_{\mathbb{R}^3} V_\ell(x) |\alpha(x)|^2 \ddd{3}x,
  \end{equation*}
  where
  $\mathcal{H}(\Gamma,\Gamma_0)$ is the relative entropy of $\Gamma$ and $\Gamma_0$. By means of 
  \cite[Lemma~3]{FHSS-micro_ginzburg_landau}, which is an extension of
   \cite[Theorem 1]{HLS2008},  giving a bound on the relative entropy 
   one obtains 
 \begin{align*}
     \mathcal{F}_T^{V_\ell}(\Gamma) -
     \mathcal{F}_T^{V_\ell}(\Gamma_0) &\geq \left\langle
     {\alpha} \left| p^2+V_\ell-\mu \right| {\alpha} \right\rangle \\
     &\quad +  \frac{1}{2}\int_{\mathbb{R}^3} (1+p^2)(\hat{\gamma} -
     \hat{\gamma}_0)^2 \ddd{3}p -C
   \end{align*}
   for an appropriate constant $C$. 
   By means of a Birman-Schwinger type argument one can further show that
   $$ p^2 + V_\ell \geq |p|^b - C_1\,,$$
   uniformly in $\ell$ for $0\leq b < 1$ and an appropriate $C_1$,
   which,  together with the constraint $|\hat \alpha  |^2 \leq 1$, then implies the statement.

\end{proof}

It was shown in \cite[Theorem 1]{HHSS} that the functional $\F^{V_\ell}_T$ attains a minimizer $(\gamma_\ell,\alpha_\ell)$ for each $V_\ell$. 
Lemma \ref{lemma:minimizer}, with $\hat \alpha_\ell (p) \leq 1$,  immediately tells us that the terms $$\int_{\mathbb{R}^3} (1 + |p|^b) |\hat \alpha_\ell(p)| \ddd{3} p \qquad {\rm and} \qquad \int_{\mathbb{R}^3} (1+|p|^2) (\hat \gamma_\ell - \hat \gamma_0)^2 \ddd{3} p$$ are uniformly bounded in $\ell$.
Let us further  mention the following useful relations between 
$\Delta_\ell$ and the minimizer $(\gamma_\ell, \alpha_\ell)$,  
  \begin{equation}
    \label{eq:el_gamma}
    \hat{\gamma_\ell} = \frac{1}{2} - \frac{1}{2}
    \frac{p^2-\mu}{K_{T,\mu}^{\Delta_\ell}},\qquad
    \Delta_\ell = 2K_{T,\mu}^{\Delta_\ell} \hat{\alpha_\ell},
  \end{equation}
  which follow from the corresponding Euler-Lagrange equations. 
  One immediate consequence of these relations and Lemma \ref{lemma:minimizer} is the uniform boundedness of 
$\int_{\mathbb{R}^3}  \hat{\gamma}_\ell(p) |p|^b \ddd{3}p$, 
which implies 
\begin{equation}\label{behgamma}
\lim_{R\to \infty} \lim_{\ell \to 0} \int_{|p|^2 \geq R} \hat \gamma_\ell(p)  \ddd{3}p = 0.
\end{equation}
 In the following lemma we show that, as $\ell\to 0$, pointwise limits
for the main quantities exist.
To this aim we introduce the notation
\begin{equation*}
  m_\mu^{\Delta_\ell}(T) =
  \frac{1}{(2\pi)^3}\int_{\mathbb{R}^3}\left(\frac{1}{K_{T,\mu}^{\Delta_\ell}}
  -\frac{1}{p^2}\right) \ddd{3}p\,.
\end{equation*}

\begin{lemma}
  \label{lemma:convergence}
  Let $(\gamma_\ell,\alpha_\ell)$ be a sequence of minimizers of
  $\mathcal{F}_T^{V_\ell}$ and
  $\Delta_\ell = \frac{2}{(2\pi)^{3/2}} \hat{V}_\ell * \hat{\alpha}_\ell$. Then there is a subsequence of $\Delta_\ell$, which we continue to denote by $\Delta_\ell$, 
 and a
  $\Delta \in \mathbb{R}_+$ such that
  \begin{enumerate}[label=(\roman*)]
  \item $|\Delta_\ell(p)|$ converges pointwise to the constant
    function $\Delta$ as $\ell\to 0$,
    \label{lemma:convergence:Delta}
  \item $\displaystyle\lim_{\ell\to 0}
    m_\mu^{\Delta_\ell}(T) = m_\mu^{\Delta}(T)$. \label{lemma:convergence:m}
  \end{enumerate}
\end{lemma}

We shall see later that it is not necessary to restrict to a subsequence, the result holds in fact for the whole sequence. 

\begin{proof}
  \proofparagraph{\ref{lemma:convergence:Delta}}
  Set $c_\ell = \frac{1}{(2\pi)^{3/2}}\int_{\mathbb{R}^3} V_\ell(x)
  \alpha_\ell(x) \ddd{3}x$. Then
  \begin{align*}
    |\Delta_\ell(p)-c_\ell|
    &\leq \frac{1}{(2\pi)^{3/2}}\int_{\mathbb{R}^3} \big|(e^{-ip\cdot
      x}-1) V_\ell(x) \alpha_\ell(x) \big| \ddd{3}x\\
    &\leq (2\pi)^{-3/2} \|\alpha_\ell\|_2
    \left(\int_{\mathbb{R}^3} \big|(e^{-ip\cdot
      x}-1) V_\ell(x)\big|^2 \ddd{3}x\right)^{1/2}.
  \end{align*}
  Now $\|\alpha_\ell\|_2$ is uniformly bounded in $\ell$ and
  $|\cdot|V \in L^2(\mathbb{R}^3)$  by assumption, so
  \begin{equation}
    \label{eq:Delta_bound}
    \begin{split}
      \int_{\mathbb{R}^3} \big|(e^{-ip\cdot x}-1) V_\ell(x)\big|^2
      \ddd{3}x &= \ell^{-1} \lambda(\ell)^2 \int_{\mathbb{R}^3}
      \big|(e^{-i \ell p \cdot x}-1) V(x)\big|^2
      \ddd{3}x\\
      &\leq \ell \lambda(\ell)^2 |p|^2\big\|\,|\cdot|V\big\|_2^2.
    \end{split}
  \end{equation}
  Hence, $|\Delta_\ell(p)-c_\ell|$ converges to zero pointwise.
  Since $\hat{\alpha}_\ell = -2
  (K_{T,\mu}^{\Delta_\ell})^{-1} \Delta_\ell$ is uniformly bounded
  in $L^2$, it is straightforward to see, using $E_\mu^{\Delta_\ell}\geq |\Delta_\ell|$, that the same holds for the sequence  $\Delta_\ell/E_\mu^{\Delta_\ell}$. 
  This fact can now be used to show that the sequence $|c_\ell|$ is a uniformly bounded. 
    Assume on the contrary that $\bar{c} = \lim \sup_{\ell \to 0} |c_\ell| =  \infty$. Then by dominated convergence
    \begin{equation}
      \label{eq:lim_P}
      \limsup_{\ell\to 0}\int_{|p|\leq R}
      \frac{|\Delta_\ell|^2}{(p^2-\mu)^2+|\Delta_\ell|^2} \ddd{3}p
      =
      \int_{|p|\leq R} \limsup_{\ell\to
        0}\frac{1}{\frac{(p^2-\mu)^2}{|c_\ell|^2}+1} \ddd{3}p
      = \frac{4}{3}\pi R^3.
    \end{equation}
  However,  the divergence of right side of Eq.~\eqref{eq:lim_P} as $R\to
  \infty$ contradicts the uniform boundedness of
  $\Delta_\ell/E_\mu^{\Delta_\ell}$  in $L^2(\mathbb{R}^3)$. Hence $\bar{c} < \infty$ and $\lim_{\ell\to 0}
  |\Delta_\ell(p)| = \bar c$ for a suitable subsequence.

  \proofparagraph{\ref{lemma:convergence:m}} Obviously by \ref{lemma:convergence:Delta} the integrand of $m_\mu^{\Delta_\ell} (T) $ 
  converges pointwise to the integrand of $m_\mu^{\Delta} (T) $. 
  By \eqref{eq:el_gamma} we are able to rewrite the integrand as 
  \begin{equation}\label{KT}  \frac{1}{K_{T,\mu}^{\Delta_\ell}(p)} - \frac{1}{p^2} = \frac 1{p^2 - \mu} - \frac 1{p^2} - \frac{2 \hat \gamma_\ell(p)}{p^2 - \mu} .\end{equation}
  Using \eqref{behgamma} we now conclude that 
  $$ \lim_{R\to \infty} \lim_{\ell \to 0} \int_{|p|^2 \geq R}\left( \frac{1}{K_{T,\mu}^{\Delta_\ell}(p)} - \frac{1}{p^2}\right ) \ddd{3} p  =0,$$
  which together with the dominated convergence inside $|p|^2 \leq R$, implies the statement of \ref{lemma:convergence:m}.

\end{proof}

\begin{proposition} Let $a = \lim_{\ell \to 0} a(V_\ell)$, then
  \label{prop:gap_eff}
  \begin{equation}
    \label{eq:gap_eff:lim}
    \lim_{\ell\to 0} \frac{1}{(2\pi)^3}\int_{\mathbb{R}^3}\left(\frac{1}{K_{T,\mu}^{\Delta_\ell}}
  -\frac{1}{p^2}\right) \ddd{3}p =- \frac{1}{4\pi a} .
\end{equation}
\end{proposition}
\begin{proof}
  We again follow the proof of \cite[Theorem~2]{BHS}. 
 Observe that with help of the second relation in \eqref{eq:el_gamma} the BCS gap equation \eqref{bcseintro} for $\alpha_\ell$ 
 can be conveniently written in the form 
  \begin{equation*}
    (K_{T,\mu}^{\Delta_\ell} + V_\ell) {\alpha}_\ell =
    0, \quad
    \textrm{ with } {\alpha}_\ell \in H^1(\mathbb{R}^3)\,.
  \end{equation*}
  By means of the Birman--Schwinger principle one concludes that $ K_{T,\mu}^{\Delta_\ell} + V_\ell$ having $0$ as an eigenvalue is equivalent to  $V_\ell^{1/2}\frac{1}{K_{T,\mu}^{\Delta_\ell}} |V_\ell|^{1/2}$
  having $-1$ as eigenvalue.
  
  We now rewrite
  $V_\ell^{1/2}\frac{1}{K_{T,\mu}^{\Delta_\ell}}
  |V_\ell|^{1/2}$ as
  \begin{align}\label{V}
    V_\ell^{1/2} \frac{1}{K_{T,\mu}^{\Delta_\ell}}
    |V_\ell|^{1/2} = V_\ell^{1/2} \frac{1}{p^2} |V_\ell|^{1/2} +
    m_\mu^{\Delta_\ell}(T) |V_\ell^{1/2}\rangle \langle
    |V_\ell|^{1/2}| + A_{\mu,T,\ell},
  \end{align}
  where $A_{\mu,T,\ell}$ is given in terms of the integral kernel
\begin{equation}
  \label{axy}
    A_{\mu,T,\ell}(x,y) =
    \frac{V_\ell(x)^{\frac{1}{2}}|V_\ell(y)|^{\frac{1}{2}}}
    {(2\pi)^3}\int_{\mathbb{R}^3}\left(
      \frac{1}{K_{T,\mu}^{\Delta_\ell}}-\frac{1}{p^2}\right)
    (\ee^{-i(x-y)\cdot p}-1) \ddd{3} p\,,
 \end{equation}
which can  be estimated, e. g., by
\begin{equation}
  \label{eq:A}
  |A_{\mu,T,\ell}(x,y)|
  \leq
  \frac{|V_\ell(x)|^{\frac{1}{2}}|V_\ell(y)|^{\frac{1}{2}}}
  {(2\pi)^3}\int_{\mathbb{R}^3}\left|
    \frac{1}{K_{T,\mu}^{\Delta_\ell}}-\frac{1}{p^2}\right|
  (|x-y|\ |p|)^{1/2} \ddd{3} p.
\end{equation}
Using the relation \eqref{KT} as well as the uniform boundedness of 
$\int_{\mathbb{R}^3}  \hat{\gamma}_\ell(p) |p|^{1/2} \ddd{3}p$ we are able to conclude that  the integral
\begin{equation*}
  \int_{\mathbb{R}^3}\left|
    \frac{1}{K_{T,\mu}^{\Delta_\ell}}-\frac{1}{p^2}\right|
  |p|^{1/2} \ddd{3} p =   \int_{\mathbb{R}^3}\left|\frac 1{p^2 - \mu} - \frac 1{p^2} - \frac{2 \hat \gamma_\ell(p)}{p^2 - \mu} \right|
  |p|^{1/2} \ddd{3} p
\end{equation*}
is uniformly bounded in $\ell$. We can thus bound the
Hilbert-Schmidt norm of $A_{\mu,T,\ell}$ by
\begin{equation}
  \|A_{\mu,T,\ell}\|_2 \leq \textrm{const}\, \|V_\ell |\cdot|\|_1^{1/2} \|V_\ell\|^{1/2}_1 
  \leq O(\ell^{3/2})
  \,. \label{eq:A_q_bound}
\end{equation}
 
  We further proceed with equation \eqref{V}. 
  By construction,
  $1+V_\ell^{1/2} \frac{1}{p^2} |V_\ell|^{1/2}$ is invertible and thus
  can be factored out, i.e.,
  \begin{align*}
    1 + V_\ell^{1/2}\frac{1}{K_{T,\mu}^{\Delta_\ell}}
    |V_\ell|^{1/2}
    = &\left(1+V_\ell^{1/2} \frac{1}{p^2} |V_\ell|^{1/2} \right) \times\\
    &\times\left[1+\tfrac{1}{1+ V_\ell^{1/2} \frac{1}{p^2}
        |V_\ell|^{1/2}}
      \left(m_\mu^{\Delta_\ell}(T)|{\scriptstyle V_\ell^{\frac{1}{2}}}\rangle
        \langle {\scriptstyle |V_\ell|^{\frac{1}{2}}}| + A_{T,\mu,\ell}\right)\right],
  \end{align*}
  with the second term on the right hand side necessarily having an eigenvalue $0$. 
  With $J= V_\ell(x) / |V_\ell(x)|$ and $X= |V_\ell|^{1/2}
    \frac{1}{p^2}|V_\ell|^{1/2} ,$  we are able to rewrite 
  $$\frac 1{1+ V_\ell^{1/2}
    \frac{1}{p^2}|V_\ell|^{1/2} } = \frac{1}{1+ JX}  = 1 - JX^{1/2}
    \frac{1}{1+X^{1/2} J X^{1/2}} X^{1/2}\,.$$
  This  allows us to bound 
   $$  \Bigl\| \frac{1}{1+ V_\ell^{1/2}
    \frac{1}{p^2}|V_\ell|^{1/2}} \Bigr\| \leq 1 + \left \|X\right \|\left \| \frac{1}{1+X^{1/2} J X^{1/2} }\right \| \leq O(\ell^{-1}),$$
  where we have used that, due to the HLS-inequality,  $ \left \|X \right \|  \leq C \|V_\ell\|_{3/2}$,  as well as  the fact that $1 +  X^{1/2} J X^{1/2}$ is self-adjoint with its lowest eigenvalue of order  $O(\ell)$. Indeed, $X^{1/2} J X^{1/2}$ has the same spectrum as $JX$.  Hence,
      \begin{align*}
  \Bigl\| \tfrac{1}{1+ V_\ell^{1/2}
    \frac{1}{p^2}|V_\ell|^{1/2}} A_{\mu,T,\ell} \Bigr\|
  &\leq
  \Bigl\| \tfrac{1}{1+ V_\ell^{1/2}
    \frac{1}{p^2}|V_\ell|^{1/2}} \Bigr\| \| A_{\mu,T,\ell} \| \leq O(\ell^{1/2}).
\end{align*}
Since 
  \begin{equation*}
    \frac{1}{1+ V_\ell^{1/2} \frac{1}{p^2}
      |V_\ell|^{1/2}}\left(m_\mu^{\Delta_\ell}(T)
      |V_\ell^{1/2}\rangle \langle |V_\ell|^{1/2}| +
      A_{T,\mu,\ell}\right)
  \end{equation*}
has an eigenvalue $-1$ and  $1+(1+V_\ell^{1/2}
  \frac{1}{p^2}|V_\ell|^{1/2})^{-1} A_{\mu,T,\ell}$ is
  invertible for small enough $\ell$, 
  we can argue  by factoring out the term  $1+(1+V_\ell^{1/2}
  \frac{1}{p^2}|V_\ell|^{1/2})^{-1} A_{\mu,T,\ell}$
 that the rank one operator
  \begin{equation*}
    m_\mu^{\Delta_\ell}(T) \left(1+  \frac{1}{1+ V_\ell^{1/2}
        \frac{1}{p^2}|V_\ell|^{1/2}} A_{\mu,T,\ell}\right)^{-1}   \frac{1}{1+ V_\ell^{1/2} \frac{1}{p^2}
      |V_\ell|^{1/2}}|V_\ell^{1/2}\rangle \langle |V_\ell|^{1/2}|
  \end{equation*}
  has an eigenvalue $-1$, which, by taking the trace, implies 
  \begin{equation}
    \label{trw}
    -1 = m_\mu^{\Delta_\ell}(T)\left< |V_\ell|^{1/2}\middle| \left[1+  \tfrac{1}{1+ V_\ell^{1/2}
            \frac{1}{p^2}|V_\ell|^{1/2}} A_{\mu,T,\ell}\right]^{-1}   \tfrac{1}{1+ V_\ell^{1/2} \frac{1}{p^2}
          |V_\ell|^{1/2}}\middle|V_\ell^{1/2}\right>.
  \end{equation}
  With the aid of Eq.~\eqref{eq:a} and the resolvent identity, we can 
  rewrite Eq.~\eqref{trw} as 
  \begin{equation}
    \label{aab}
    \begin{split}
      & 4\pi a(V_\ell) + \frac 1 { m_\mu^{\Delta_\ell}(T)} \\
      &=
      \left< |V_\ell|^{1/2}\right| \tfrac{1}{1+ V_\ell^{1/2}
        \frac{1}{p^2}|V_\ell|^{1/2}} A_{\mu,T,\ell} \Bigl[1+
      \tfrac{1}{1+ V_\ell^{1/2} \frac{1}{p^2}|V_\ell|^{1/2}}
      A_{\mu,T,\ell}\Bigr]^{-1} \times\\
      &\qquad\times\tfrac{1}{1+ V_\ell^{1/2}
        \frac{1}{p^2}
        |V_\ell|^{1/2}}\left|V_\ell^{1/2}\right>, 
    \end{split}
  \end{equation}
where the right hand side is bounded by  
\begin{align*}
  \|V_\ell\|_1 \Bigl\| &\tfrac{1}{1+ V_\ell^{1/2}
    \frac{1}{p^2}|V_\ell|^{1/2}}\Bigr\| \Bigl\| \tfrac{1}{1+ V_\ell^{1/2}
    \frac{1}{p^2}|V_\ell|^{1/2}} A_{\mu,T,\ell} \Bigr\| \times \\
  &\times
  \Bigl\| \Bigl[1+
  \tfrac{1}{1+ V_\ell^{1/2} \frac{1}{p^2}|V_\ell|^{1/2}}
  A_{\mu,T,\ell}\Bigr]^{-1} \Bigr\| \leq O(\ell^{1/2}) \,.
\end{align*}
This implies Eq.~\eqref{eq:gap_eff:lim}  and completes the proof.

\end{proof}

With the aid of Lemma~\ref{lemma:convergence} and
Proposition~\ref{prop:gap_eff}, we can now finish the proof of Theorem \ref{thm:gap_eff}. 
\begin{proof}[Proof of Theorem \ref{thm:gap_eff}]
We know from Lemma~\ref{lemma:convergence} and Proposition~\ref{prop:gap_eff} that  $|\Delta_\ell(p)|$ has a subsequence that converges to a constant function $\Delta$, which satisfies 
the equation  
\begin{equation*}
    -\frac{1}{4\pi a} = 
    \frac{1}{(2\pi)^3}\int_{\mathbb{R}^3}\left(\frac{1}{K_{T,\mu}^{\Delta}}
      -\frac{1}{p^2}\right) \ddd{3}p.
\end{equation*}
Since the solution $|\Delta|$ of \eqref{eq:gap_eff}  is unique we obtain that the sequence 
$|\Delta_\ell(p)|$ converges to the unique solution of  \eqref{eq:gap_eff}.
Furthermore, this shows that the limit of $|\Delta_\ell|$ does not vanish in the case that 
$T< T_c$, and that the limit vanishes for $T\geq T_c$. 
%
\end{proof}

\end{document}